\documentclass[conference,a4paper]{IEEEtran}

\usepackage{amsmath,amsthm,amssymb}
\usepackage[cp1251]{inputenc}
\usepackage{tikz}
\usepackage{graphicx}
\usepackage{multirow}
\usepackage{tabularx}
\usepackage{blkarray}
\usepackage{subcaption}
\DeclareMathOperator{\F}{\mathbb F}
\DeclareMathOperator{\W}{\mathcal W}

\DeclareMathOperator{\C}{ {\mathcal C}}

\newcommand{\mF}{\mathcal F}
\newcommand{\set}[1]{\left\{{#1}\right\}}

\DeclareMathOperator{\sgn}{sgn}
\newcommand{\fs}[2]{p_{#2}(#1)/s_{#2}(#1)}

\theoremstyle{note}
\newtheorem{example}{Example}
\newtheorem{remark}{Remark}
\theoremstyle{plain}

\newtheorem{lemma}{Lemma}

\begin{document}

\title{
Recursive Trellis Processing of Large Polarization Kernels
}
\author{Peter Trifonov}
\author{\IEEEauthorblockN{Peter Trifonov} 
\IEEEauthorblockA{ITMO University,
Russia\\Email: pvtrifonov@itmo.ru}}

\maketitle

\begin{abstract}
A reduced complexity algorithm is presented for computing the log-likelihood ratios arising in the successive cancellation decoder for polar codes with large kernels of arbitrary dimension. The proposed algorithm exploits recursive trellis representation of the codes generated by submatrices of the polarization kernel, and enables codes based on large kernels to provide better performance compared to the codes based on Arikan kernel with the same decoding complexity.
 \end{abstract}
\section{Introduction}
Polar codes, introduced by Arikan in \cite{arikan2009channel}, have already made their way into the 5G specification. However, their role is limited there for encoding of small data blocks, since long polar codes do not compete well with other modern codes, such as LDPC. Polar codes with large kernels were shown to have asymptotically optimal scaling exponent \cite{fazeli2018binary}.
These codes have decoding complexity $O(n\log n)$ operations of computing LLRs. However, each such operation, called kernel processing or kernel marginalization \cite{bioglio2018marginalization}, has complexity $O(2^ll)$ if implemented straightforwardly. An approximate kernel processing method based on box-and-match decoder was suggested in \cite{miloslavskaya2014sequentialBCH}.  Some kernels of dimension 16 and 32 were published together with efficient processing algorithms \cite{trofimiuk2019reduced}. It was shown that the corresponding polar codes have lower decoding complexity compared to codes based on Arikan kernel with the same performance.  However, no efficient processing techniques for generic polarization kernels have been published so far.

In this paper we present such algorithm, and report its complexity for some kernels available in the literature. We show that the proposed approach has much lower complexity compared to the Viterbi algorithm.
\section{Background}
\subsection{Polar codes}
Polar code \cite{korada2010polar} is a set of vectors $c_0^{n-1}=u_0^{n-1}K^{\otimes l}$, where $K$ is a non-singular $l\times l$ matrix called polarization kernel, $n=l^m$, $u_i=0$ for $i\in \mF$,  $\mF\subset[n]$ is a frozen set, and $[n]=\set{0,\dots,n-1}$.
This definition can be generalized to obtain mixed kernel polar codes with the codewords given by $c_0^{n-1}=u_0^{n-1}(K_{l_1}\otimes K_{l_2}\otimes \cdots\otimes K_{l_m})$ \cite{presman2016mixed,bioglio2020multikernel}, where $K_{l_i}$ is a kernel of dimension $l_i$. Unless stated otherwise, we consider here the case of all kernels being the same.

Decoding of polar codes can be implemented by the successive cancellation algorithm, which makes decisions $$\hat u_i=\begin{cases}0,&i\in \mF\\
\arg\max_{u_i\in \F_2} \W_m^{(i)}(\hat u_0^{i-1}\bullet u_i|r_0^{n-1}),&i\notin\mF,\end{cases}$$
where 
$
\W_m^{(i)}(u_0^i|r_0^{n-1})=\sum_{u_{i+1}^{n-1}}\prod_{j=0}^{n-1}W_{0}^{(0)}((u_0^{n-1}K^{\otimes l})_i|r_i),
$  $\W_0^{(0)}(c|r)=W_0^{(0)}(c|r)=\frac{W(r|c)}{2W(r)}$,  $W(r|c)$ is the channel transition probability, and $\bullet$ denotes the concatenation operator.
These probabilities can be recursively computed as 
\begin{align}
\W_m^{(li+s)}(u_0^{li+s}|r_0^{n-1})=\nonumber\quad\quad\quad\quad\quad\quad\quad\quad\quad\quad\quad\quad\quad\\
\sum_{u_{s+1}^{l-1}}\prod_{j=0}^{l-1}\W_{m-1}^{(i)}\left((u_{lt}^{l(t+1)-1}K),t\in[i+1]|r_{0,j}^{n-1}\right), \end{align}
where $r_{0,j}^{n-1}=(r_j,r_{j+l},\dots,r_{j+n-l})$. This operation is known as kernel processing or kernel marginalization \cite{trifonov2014binary,bioglio2018marginalization}.

It was suggested in \cite{miloslavskaya2014sequentialBCH} to approximate the above probabilities as
$
\W_m^{(i)}(u_0^i|r_0^{n-1})\approx W_m^{(i)}(u_0^i|r_0^{n-1})=\max_{u_{i+1}^{n-1}}\prod_{j=0}^{n-1}W_{0}^{(0)}((u_0^{n-1}K^{\otimes l})_i|r_i),
$ so that 
\begin{align}
W_m^{(li+s)}(u_0^{li+s}|r_0^{n-1})=\nonumber\quad\quad\quad\quad\quad\quad\quad\quad\quad\quad\quad\quad\quad\\
\max_{u_{s+1}^{l-1}}\prod_{j=0}^{l-1}W_{m-1}^{(i)}\left((u_{lt}^{l(t+1)-1}K),t\in[i+1]|r_{0,j}^{n-1}\right).
\label{mKernProcMax} 
\end{align}
Decoding can be implemented in the LLR domain using LLRs 
$$ S_m^{(i)}(u_0^{i-1},r_0^{n-1})=
\ln\frac{W_m^{(i)}(u_0^{i-1}\bullet0|r_0^{n-1})}{W_m^{(i)}(u_0^{i-1}\bullet1|r_0^{n-1})}.$$
Assume for the sake of simplicity  $m=1$. It can be seen that 
\begin{align}
S_1^{(i)}(u_0^{i-1},r_0^{n-1})=&\max_{u_{i+1}^{l-1}}\sum_{j=0}^{l-1}\ln W_0^{(0)}\left(((u_0^{i-1}\bullet 0\bullet u_{i+1}^{l-1})K)_j|r_j\right)\nonumber\\
-&\max_{u_{i+1}^{l-1}}\sum_{j=0}^{l-1}\ln W_0^{(0)}\left(((u_0^{i-1}\bullet 1\bullet u_{i+1}^{l-1})K)_j|r_j\right)\nonumber\\
=&\frac{1}{2}\max_{u_{i+1}^{l-1}}Q((u_0^{i-1}\bullet0\bullet u_{i+1}^{l-1})K,r_0^{l-1})\nonumber\\
-&\frac{1}{2}\max_{u_{i+1}^{l-1}}Q((u_0^{i-1}\bullet 1\bullet u_{i+1}^{l-1})K,r_0^{l-1}),
\label{mLLREW}
\end{align}
where $Q(c_0^{l-1},r_0^{l-1})=\sum_{j=0}^{l-1}(-1)^{c_j}S_0^{(0)}(r_j)$
is the correlation function.
The objective of the present paper is to provide an efficient method for computing LLRs $S_1^{(i)}(u_0^{i-1},r_0^{n-1})$.
\subsection{Recursive maximum likelihood decoding algorithm}
It can be seen that computing \eqref{mLLREW} reduces to ML\ decoding in the cosets of the code generated by $l-i-1$ last rows of matrix $K$, provided that all symbols $u_s,s\in[l]$, are equiprobable binary random values.
An efficient recursive maximum likelihood decoding algorithm for linear block codes was suggested in \cite{fujiwara1998trellisbased}.
The idea is to recursively partition the received noisy vector into a number of sections $[x,y)$, identify for each section a number of most likely vectors $c_x^{y-1}\in \F_2^{y-x}$ corresponding to the received values $r_x^{y-1}$, and recursively combine them to obtain most likely vectors for longer sections.

Given a linear code $C$, let $C_{h,h'}$ be its subcode, such that all its codewords have non-zero  symbols only in positions $h\leq i<h'$.   Let $p_{h,h'}(C)$ be a linear code obtained by puncturing all symbols, except those in positions  $h\leq i<h'$, from codewords of $C$. Let us further define $s_{h,h'}(C)=p_{h,h'}(C_{h,h'})$, i.e. a code  obtained from $C$ by shortening it on all symbols except those with indices
$h\leq i<h'$. The codes $s_{h,h'}(C)$ and $p_{h,h'}(C)$ will be referred to as {\em section codes}. Consider a minimal trellis of code $C$, and its section corresponding to symbols from $x$ to $y$. It is possible to show that the paths between two adjacent states in this section correspond to a coset in $p_{x,y}(C)/s_{x,y}(C)$ \cite{fujiwara1998trellisbased}. This coset may appear in the trellis several times. Hence, one can simplify  maximum likelihood (ML) decoding by pre-computing the metrics of these paths. That is, for each coset $D\in p_{x,y}(C)/s_{x,y}(C)$ one needs to identify the most probable element $l(D)$, and store its correlation  $Q(D)=Q(l(D),r_x^{y-1})$. Let the {\em composite branch table (CBT)}  $T_{x,y}$ be an array containing values  $T_{x,y}[v].l=l(D)$ and $T_{x,y}[v].q=Q(D),$ where $v$ is an index of $D$.   In the case of conventional  decoding of an $(n,k)$ code, $p_{0,n}(C)/s_{0,n}(C)$ contains a single element, so the corresponding CBT\ has one entry $T_{0,n}[0]$, which gives a solution of the ML  decoding problem. 
 
The straightforward approach to  construction of a composite branch table for some code $C$ is to enumerate all codewords of $p_{x,y}(C)$, and find the most probable one for each coset in $p_{x,y}(C)/s_{x,y}(C)$. We assume that this method is used for $y-x\leq 2$. However, more efficient approach was suggested in  \cite{fujiwara1998trellisbased} for the case of  $y-x\geq 2$.  Consider some $z:x<z<y$.  Let the generator matrix of $p_{x,y}(C)$ be represented as  \begin{equation}
\label{mMSF}
G^{(p)}_{x,y}=\begin{pmatrix}
G^{(s)}_{x,z}&0\\
0&G^{(s)}_{z,y}\\
G^{(00)}_{x,y}&G^{(01)}_{x,y}\\\hline
G^{(10)}_{x,y}&G^{(11)}_{x,y}\\
\end{pmatrix},
\end{equation}
where $G^{(s)}_{x,y}=\begin{pmatrix}
G^{(s)}_{x,z}&0\\
0&G^{(s)}_{z,y}\\
G^{(00)}_{x,y}&G^{(01)}_{x,y}
\end{pmatrix}$ is a generator matrix of $s_{x,y}(C)$, and $G_{x,y}^{(00)},G_{x,y}^{(01)}$ are some $k_{x,y}''\times(z-x)$ and $k_{x,y}''\times (y-z)$ matrices, respectively, where $k_{x,y}'=k_{x,y}'(C)$ and $k_{x,y}''=k_{x,y}''(C)$ are some code-dependent integers. There is an one-to-one correspondence between vectors $vG_{x,y}'$, where $G_{x,y}'=\begin{pmatrix}G^{(10)}_{x,y}&G^{(11)}_{x,y}\end{pmatrix}$ is a $k'_{x,y}\times (y-x)$ matrix, and  cosets  $D\in\fs{C}{x,y}$.  Here $k_{x,y}',k_{x,y}''$ are some integers, which depend on code structure, and can be obtained from the minimum span form of its generator matrix.

Hence, we write $T_{x,y}[v].l:=l(D)$ and $T_{x,y}[v].q=Q(D)$, with $D$ being a coset corresponding to $v$.
It can be seen that 
\begin{align}
T_{x,y}[v].q=&\max_{c_{x}^{y-1}\in D}Q(c_x^{y-1},r_x^{y-1})\nonumber\\
=&\max_{w\in \F_2^{k_{x,y}''}} \left(T_{x,z}[a].q+T_{z,y}[b].q\right),v\in \F_2^{k'_{x,y}}
\label{mCBTMax}
\end{align}
where $a$ and $b$ are indices of the cosets $D'\in\fs{C}{x,z}$ and $D''\in\fs{C}{z,y}$, respectively,  such that  $\begin{pmatrix}w&v\end{pmatrix}\begin{pmatrix}G^{(00)}_{x,y}\\G^{(10)}_{x,y}\end{pmatrix}\in D'$ and  $\begin{pmatrix}w&v\end{pmatrix}\begin{pmatrix}G^{(01)}_{x,y}\\G^{(11)}_{x,y}\end{pmatrix}\in D''$.   Such values $a,b$ can be identified from the  system of equations
\begin{align*}
\begin{pmatrix}a'&a
\end{pmatrix}\begin{pmatrix}
G_{x,z}^{(s)}\\
G_{x,z}'
\end{pmatrix}=&\begin{pmatrix}w&v\end{pmatrix}\begin{pmatrix}G^{(00)}_{x,y}\\G^{(10)}_{x,y}\end{pmatrix}\\
\begin{pmatrix}b'&b
\end{pmatrix}\begin{pmatrix}
G_{z,y}^{(s)}\\
G_{z,y}'
\end{pmatrix}=&\begin{pmatrix}w&v\end{pmatrix}\begin{pmatrix}G^{(01)}_{x,y}\\G^{(11)}_{x,y}\end{pmatrix},
\end{align*}
where $a',b'$ are some irrelevant values.
Obviously, the solutions are given by $a=\begin{pmatrix}w&v\end{pmatrix}\widehat G_{x,y}$ and $b=\begin{pmatrix}w&v\end{pmatrix}\widetilde G_{x,y}$ for some matrices $\widehat G_{x,y}$ and $\widetilde G_{x,y}$. The corresponding most likely coset representatives are given by $T_{x,y}[v].l=T_{x,z}[\hat a].l\bullet  T_{x,z}[\hat b].l$, where $\hat a,\hat b$ are the values of $a$ and $b$, which deliver maximum in \eqref{mCBTMax}.

The complexity of this calculation is $O(2^{k_{x,y}'+k_{x,y}''})$. It can be further reduced by exploiting the  tricks suggested in \cite{fujiwara1998trellisbased}. The overall decoding complexity strongly depends on the sectionalization method being used, i.e. a rule for selection of the partitioning point $z$ for some $x,y$. This approach, known as recursive maximum likelihood decoding (RMLD), was shown to be  more efficient compared to the Viterbi algorithm \cite{fujiwara1998trellisbased}. 
\section{Recursive trellis processing}
\subsection{Extended kernel codes}
Polar codes can be considered as a result of recursive application of the construction of generalized concatenated codes \cite{blokh1974coding,trifonov2012efficient}.
These codes rely on non-systematic inner codes, as well as the corresponding soft-decision decoding algorithms. An optimal soft-input soft-output decoding algorithm for non-systematically encoded linear block codes was presented in \cite{griesser2002aposteriori}. This algorithm can be easily tailored to implement computation of \eqref{mLLREW}. To do this, let us consider an extended $(l+1,l-i)$ code $\overline {\mathcal C}^{(i)}$ generated by matrix $G^{(i)}$. First $l$ columns of this matrix are obtained by taking $l-i$ last rows of kernel $K$. The last column has $1$ in the $0$-th row, and zeroes in the remaining positions. Assuming that $u_0^{i-1}=0$, one obtains that computing \eqref{mLLREW}  is equivalent to finding the most probable codewords of code $\overline {\mathcal C}^{(i)}$ having $0$ and $1$  in the last symbol. This can be implemented by running the Viterbi algorithm over the trellis of  $\overline {\mathcal C}^{(i)}$, assuming that the last codeword symbol is erased. The same trellises, although with  different labeling, can be used to implement decoding in the cosets of the extended codes, which arise in the case of $u_0^{i-1}\neq 0$.
\begin{example}
Consider  Arikan kernel $F_2=\begin{pmatrix}
1&0\\1&1
\end{pmatrix}$. One obtains $G^{(0)}=\begin{pmatrix}
1&0&1\\1&1&0
\end{pmatrix}$ and $G^{(0)}=\begin{pmatrix}
1&1&1
\end{pmatrix}$. 
\end{example}
 
%
%
%
%
%
%
%
%
%
%
\subsection{Recursive processing of polarization kernels}
We propose to compute \eqref{mLLREW} by applying the RMLD algorithm to the cosets of the extended kernel codes. Let $\overline \C(u_0^{i-1})=w+\overline \C^{(i)}$ be the coset of $\C^{(i)}$ given by prior decisions $u_0^{i-1}$, where $w=(u_0^{i-1}K_{0,\dots,i-1},0)$, and $K_{0,\dots,i-1}$ is the matrix consisting of $i$ top rows of $K$. For any section $[x,y)$, the cosets associated with states in the recursive trellis for $\overline \C(u_0^{i-1})$ are obtained from those for $\overline \C^{(i)}$  as $D(u_0^{i-1})=\set{f+w_{x}^{y-1}|f\in D}, D\in\fs{\overline \C^{(i)}}{x,y}$.

It can be seen that $\fs{\overline \C^{(i)}}{0,l}$
contains two cosets, which correspond to $u_i=0$ and $u_i=1$.
Hence, 
\begin{equation}
\label{mFinalLLR}
S_1^{(i)}(u_0^{i-1},r_0^{l-1})=\frac{T_{0,l}[0].q-T_{0,l}[1].q}{2},
\end{equation}
 where $T_{0,l}$ is the composite branch table constructed for $\overline\C(u_0^{i-1})$  given noisy vector $r_0^{l-1}$.  It is assumed in what follows that $CBT$ entries contain only $q$ values, and the corresponding coset representatives $l(D)$ are omitted.

Furthermore, we propose to reuse the composite branch tables, or their parts, obtained at successive phases $i$. To do this, we need to identify how CBTs evolve with $i$, find a way to handle prior decisions $u_0^{i-1}$, design efficient algorithms for construction of CBTs for short sections, and obtain an optimal sectionalization strategy.
\subsubsection{Reusing the CBTs}
Assume that the same sectionalization is used for all phases $i$.  
Obviously, $p_{x,y}(\overline C^{(i+1)})\subset p_{x,y}(\overline C^{(i)})$ and $s_{x,y}(\overline C^{(i+1)})\subset s_{x,y}(\overline C^{(i)}), i\in[l-1]$
for any $x,y$, such that $0\leq x<y\leq l$. 
\begin{lemma}
If $p_{x,y}(\overline C^{(i+1)})=p_{x,y}(\overline C^{(i)})$ and 
$s_{x,y}(\overline C^{(i+1)})=s_{x,y}(\overline C^{(i)})$, then for any $ u_i\in\F_2$ the composite branch table $T_{x,y}$ constructed for $\overline\C(u_0^{i-1})$ is identical to the one constructed for $\overline\C(u_0^{i})$, denoted $T_{x,y}'$, for the same received vector $r_0^{l-1}$.
\end{lemma}
\begin{proof}
Since both section codes are identical for phases $i$ and $i+1$, the CBTs have the same size.
$p_{x,y}(\overline C^{(i+1)})=p_{x,y}(\overline C^{(i)})$ implies that $\kappa=(K_{i,x},\dots,K_{i,y-1})\in p_{x,y}(\overline C^{(i+1)})$.  Hence, for any $u_i\in\F_2$ both $\fs{\overline \C^{(i)}}{x,y}$ and $\fs{\overline \C^{(i+1)}}{x,y}$ have cosets containing $u_i\kappa,u_i\in \F_2$. This implies that both $T_{x,y}$ and $T_{x,y}'$ contain the same values.
\end{proof}
The above lemma suggests that one does not need to recompute the CBTs for those sections, where section codes do not change from phase $i$ to phase $i+1$. 

In what follows, we assume that  
\begin{equation}
\label{sSubsectionNoShortChange}
s_{x,z}(\C^{(i+1)})=s_{x,z}(\C^{(i)}) \text{\,and\,} s_{z,y}(\C^{(i+1)})=s_{z,y}(\C^{(i)}).
\end{equation}  
Let us  assume temporarily that $u_0^{i-1}=0$.

Even if section codes do change, it is still possible to reuse some results obtained at prior phases. Let $k_{i,x,y}'=k_{x,y}'(\C^{(i)})$ and $k_{i,x,y}''=k_{x,y}''(\C^{(i)})$. First, observe that if $k_{i+1,x,y}''=k_{i,x,y}''$,
but $k_{i+1,x,y}'<k_{i,x,y}'$, then $\fs{\C^{(i+1)}}{x,y}\subset \fs{\C^{(i)}}{x,y}$, so that the corresponding CBT at phase $i+1$ can be obtained as a subvector of the CBT at phase $i$.

Second, we propose to implement maximization recursively, and keep all intermediate results. More specifically, we propose to rewrite 
\eqref{mCBTMax} as 
\begin{align}
T_{x,y}[v].q=&\max_{w\in \F_2^{k_{i,x,y}''}} \left(T_{x,z}[a].q+T_{z,y}[b].q\right)\nonumber\\
=&\max_{w_{k_{i,x,y}''-1}}\dots \max_{w_1}\max_{w_0}\left(T_{x,z}[a].q+T_{z,y}[b].q\right).
\label{mCBTMinSuccessive}
\end{align}
Instead of storing in the CBT\ the final results of maximization in \eqref{mCBTMinSuccessive}, we propose to keep the intermediate results of maximization for all $w$. These values can be arranged in a binary tree for each $v\in \F_2^{k'_{x,y}}$, so that a path from a root in this tree can be specified by values $w_{k_{i,x,y}''-1},w_{k_{i,x,y}''-2},\dots,w_0$. By maximization forest we denote the set of such trees obtained at some phase for a given section.  The subtrees within a forest can be indexed by variables $w,v$.

We propose to use the above described maximization tree constructed at some phase $i_0$ to obtain CBTs for all $i\geq i_0$, where \eqref{sSubsectionNoShortChange} holds. Let $i_1>i_0$ be the smallest integer, where this does not hold.
 \begin{lemma}
\label{lReduceShortened}
Let $G_{j,x,y}''=\begin{pmatrix}
G^{(00)}_{x,y}&G^{(01)}_{x,y}
\end{pmatrix}$ and $G_{j,x,y}'=\begin{pmatrix}
G^{(10)}_{x,y}&G^{(11)}_{x,y}
\end{pmatrix}$ 
be the matrices obtained from \eqref{mMSF} for code $\C^{(j)}$ for any $j$.
 If all matrices $G_{i,x,y}''$ are nested, so that 
 $G_{i+1,x,y}''$ occupies top rows of $G_{i,x,y}''$ for any $i:i_0\leq i< i_1$, then the maximization forest for phase $i:i_{0}\leq i\leq i_1$ can be obtained by taking the subtrees of the trees in the forest constructed at phase $i_0$,  given by values $w$ and $v$ satisfying the equation
\begin{equation}
\label{mSubforestReduction}
 \begin{pmatrix}
w&v
\end{pmatrix}\begin{pmatrix}
G_{i_0,x,y}''\\
G_{i_0,x,y}'
\end{pmatrix}=
\begin{pmatrix}
\overline w&\overline v
\end{pmatrix}\begin{pmatrix}
G_{i,x,y}''\\
G_{i,x,y}'
\end{pmatrix},
\end{equation}
where $\overline w\in \F_2^{k_i'',x,y},\overline v\in \F_2^{k_i',x,y}$ denote the subtree indices in the forest at phase $i$.
\end{lemma}
\begin{proof}
Assumption \eqref{sSubsectionNoShortChange} ensures that the maximization trees for all $i$ are obtained from the same values.

Observe that the matrices in \eqref{mSubforestReduction} together with the generator matrices of $s_{x,z}(\C^{(i_0)})$ and $s_{z,y}(\C^{(i_0)})$ constitute generator matrices of $p_{x,y}(\C^{(i_0)})$ and $p_{x,y}(\C^{(i)})$, given by \eqref{mMSF}. Since $p_{x,y}(\C^{(i)})\subset p_{x,y}(\C^{(i_0)})$, for any $\overline w,\overline v$ there is a unique solution $\begin{pmatrix}
 w&v
\end{pmatrix}=\begin{pmatrix}
\overline w&\overline v
\end{pmatrix}M_{i,x,y}$ of \eqref{mSubforestReduction}, where $M_{i,x,y}$ is a matrix, which can be constructed at the design time.
 Since $G_{i,x,y}$ occupies $k_{i,x,y}''$ top rows of $G_{i_0,x,y}''$, one obtains $w_j=\overline w_j,0\leq j< k_{i,x,y}''$. Hence, for each $\overline v$ the corresponding maximization tree indeed appears as a subtree indexed by $w_j$ of the $v$-th maximization tree.
\end{proof}
\begin{remark}
Observe that it is always possible to obtain  $G_{j,x,y}''$ in the form required by Lemma  \ref{lReduceShortened} by applying elementary row operations.
\end{remark}

\subsubsection{Handling prior decisions}
The SC decoder needs to take into account at phase $i$ the values $u_0^{i-1}$.
This reduces to decoding in cosets of section codes. The corresponding coset representative for section $[x,y)$ can be computed as a linear combination of subvectors  $\begin{pmatrix}
K_{j,x}&K_{j,x+1}&\dots&K_{j,y-1}
\end{pmatrix}, 0\leq j<i$, of the kernel.

 Consider first
the case of section $[x,y)$, where the maximization forest \eqref{mCBTMinSuccessive} is constructed from scratch. If there exists a solution of 
\begin{equation}
\begin{pmatrix}
f_{j,x,y}&h_{j,x,y}
\end{pmatrix}\begin{pmatrix}
G^{(s)}_{i,x,y}\\
G_{i,x,y}'
\end{pmatrix}=\begin{pmatrix}
K_{j,x}&K_{j,x+1}&\dots&K_{j,y-1}
\end{pmatrix},
\end{equation}
then the CBT for section $[x,y)$ contains at position $h_{j,x,y}$ the required coset representative. In this case we assume $h_{j,x',y'}=0$ for all $x<x'<y'<y$.
Otherwise, we assume $h_{j,x,y}=0$.  Given a vector of prior decisions $u_0^{i-1}$, we obtain the corresponding position offset at section $[x,y)$ as $h_{x,y}=\sum_{j=0}^{i-1}u_jh_{j,x,y}$, so that \eqref{mCBTMinSuccessive} becomes 
\begin{align}
\label{mCBTOffset}
T_{x,y}[v].q=\max_{w} \left(T_{x,z}[a+h_{x,z}].q+T_{z,y}[b+h_{z,y}].q\right).
\end{align}
The values  $h_{x,y}$  are similar to partial sums, which arise in the SC decoder of Arikan polar codes.

For those sections, where the CBT is obtained by taking subtrees of \eqref{mCBTMinSuccessive},
we need to check if there is a solution of 
\begin{equation}
\omega_{j,x,y}\begin{pmatrix}
G^{(s)}_{i_0,x,y}\\
G_{i_0,x,y}'
\end{pmatrix}=\begin{pmatrix}
K_{j,x}&K_{j,x+1}&\dots&K_{j,y-1}
\end{pmatrix}.
\end{equation}
If this equation does not have a solution for some $j<i$, the corresponding coset representatives have already been accounted for at smaller sections while constructing the CBT for section $[x,y)$ at phase $i_0$. and we assume $\omega_{j,x,y}=0$. Otherwise, the corresponding CBT\ can be obtained from the maximization forest constructed at phase $i_0<i$ by taking entries 
$$\begin{pmatrix}
 w&v
\end{pmatrix}=\begin{pmatrix}
\overline w&\overline v
\end{pmatrix}M_{i,x,y}+\sum_{j=0}^{i-1}u_i\omega_{j,x,y}.$$

\subsubsection{Special trellises}
It can be seen  that the result of \eqref{mFinalLLR} does not change if the same value is subtracted from all CBT entries at any section. Doing this enables one to compute some  CBTs with smaller number of operations compared to \eqref{mCBTOffset}.   The following special cases were identified:
\begin{enumerate}
\item $k_{i,x,y}'=k_{i,x,y}''=1,\hat G_{i,x,y}=\begin{pmatrix}1\\0\end{pmatrix},\tilde G_{i,x,y}=\begin{pmatrix}1\\1\end{pmatrix}$. Let $a=\frac{T_{x,z}[0].q-T_{x,z}[1].q}{2}$, $b=\frac{T_{z,y}[0].q-T_{z,y}[1].q}{2}$. We propose to set $T_{x,y}[0].q=\sgn(a)\sgn(b)\min(|a|,|b|)$ and $T_{x,y}[1].q=-T_{x,y}[0].q$.    If one can guarantee that $T_{x,z}[1].q=-T_{x,z}[0].q$, then it is possible to further simplify computation by setting $a=T_{x,z}[0].q$. Similar simplification applies to $b$.  This trick allows one to construct section CBT using just 1 comparison operation, instead of 2 comparison and 4 summations for a straightforward implementation.
\item $k_{i,x,y}'=1,k_{i,x,y}''=0,\hat G_{i,x,y}=\tilde G_{i,x,y}=(1)$. Using the same definitions as above, one obtains $T_{x,y}[0].q=(-1)^{h_{x,z}}a+(-1)^{h_{z,y}}b,T_{x,y}[1].q=-T_{x,y}[0].q$. \item $k_{i,x,y}'=0,k_{i,x,y}''=2,\hat G_{i,x,y}=\begin{pmatrix}1\\0\end{pmatrix},\tilde G_{i,x,y}=\begin{pmatrix}0\\1\end{pmatrix}$. Using the same definitions as above, we propose to set $T_{x,y}[0].q=a+b,T_{x,y}[1].q=b-a,T_{x,y}[2].q=-T_{x,y}[1].q,T_{x,y}[3].q=-a$.
\item If there is all-1 row in $G_{i,x,y}^{(s)}$, then one step of maximization can be avoided in \eqref{mCBTMinSuccessive} by taking the absolute values of the corresponding terms. This reduces the complexity of construction of the maximization tree by a factor of 2.
\end{enumerate}
Simplification tricks 1 and 2 together with the results in 
 \cite{trifonov2019trellis} fully establish the equivalence of the proposed approach and the min-sum SC algorithm for the case of   $K=B_\mu\begin{pmatrix}1&0\\1&1\end{pmatrix}^{\otimes \mu}$, where $B_\mu$ is the bit reversal permutation matrix.

\subsubsection{Optimal sectionalization}
The total complexity of kernel processing is equal to $\mathbf C=\delta_i+\sum_{i=0}^{l-1} c_{i,0,l},$ where $\delta_i$ is the complexity of computing the final LLR from the obtained CBT, and $c_{i,x,y}$ is the complexity of construction of the CBT for section $[x,y)$ at phase $i$. In most cases the former operation reduces to computing \eqref{mFinalLLR}, i.e. costs 1 subtraction. However, if above described special trellises 1 or 2 arise at section $[0,l)$, then $\delta_i=0$.
Furthermore, $$c_{i,x,y}=\begin{cases}
m_{i,x,y},\text{\, if CBTs for subsections can be reused}\\
m_{i,x,y}+c_{i,x,z}+c_{i,z,y},\text{\,otherwise,}
\end{cases}$$
where $$m_{ixy}=\begin{cases}
0,\text{\,\,\,\,\,\, if forest reuse is possible}\\
M_j,\text{\, if type-$j$ special trellis is encountered,}\\
2^{k_{ixy}'+k_{ixy}''-f_{ixy}}+2^{k_{ixy}'}(2^{k_{ixy}''-f_{ixy}}-1),\text{otherwise,}
\end{cases}
$$   
$f_{ixy}$  is equal 1 if the above described 4-th simplification trick is applicable, and 0 otherwise, and $M_j$ is the complexity of type-$j$ special trellis, $1\leq j\leq 3$.  The first and second terms in the latter expression represent the number of summations and  comparisons, respectively. 
These expressions can be used in the optimization algorithm given in \cite{fujiwara1998trellisbased} to obtain  splitting position $z$ for each section $[x,y)$, so that the overall complexity is minimized. Observe that  sectionalization should be optimized jointly for  codes $\overline \C^{(i)}$ at all phases $i$.
\section{Numeric results}
Table \ref{tKernComp} presents processing complexity for some polarization kernels. Here  $K_l$ denotes a kernel of size $l$.
\begin{table}
\caption{Kernel processing complexity}
\label{tKernComp}
\setlength{\tabcolsep}{2pt}
\begin{tabular}{|c|>{$}c<{$}|>{$}c<{$}|c|c|c|c|c|c|}\hline
\multirow{2}{*}{Kernel $K_l$}&\multirow{2}{*}{$E(K_l)$}&\multirow{2}{*}{$\mu(K_l)$}&\multicolumn{3}{c|}{State of the art}&\multicolumn{2}{c|}{Proposed}\\\cline{4-8}
&&&Method&Add&Comp.&Add&Comp.\\\hline
$K_{16}B_4$ \cite{trofimiuk2021window}&0.51828&3.45 &window&95&86&131&105\\\hline
$K_{32}B_5$ \cite{trofimiuk2021window}&0.521936&3.417 &window&297&274&406&262\\\hline
$K_{32}^r$\cite{trofimiuk2021searchArXiv}&0.52194&3.42111&Viterbi&4536 &9072&355&191\\\hline
$K_{32}^{bch}$ \cite{moskovskaya2020design}&0.53656&3.1221&Viterbi&99745&199490&31079 &  28337\\\hline
$K_{32}^{enbch'}$ \cite{abbasi2020large}&0.53656&3.1221&window&\multicolumn{2}{c|}{2864420}&  32183   &29873\\\hline
$K_{20}^\ast$ \cite{trofimiuk2021searchArXiv}&0.506169&3.43827&Viterbi&7524&15054&2893  &  2001 \\\hline
$K_{20}$ \cite{trofimiuk2021searchArXiv}&0.49943&  3.64931&Viterbi&1866 &3756&289    & 189 \\\hline
$K_{24}^\ast$ \cite{trofimiuk2021searchArXiv}& 0.51577&3.3113&Viterbi&9922&19860&1621    &1207\\\hline
$K_{24}$ \cite{trofimiuk2021searchArXiv}&0.502911&3.61903&Viterbi&2102&4218&241    &124\\\hline
\end{tabular}
\end{table}
For comparison, we present the complexity of the window-based and Viterbi  processing algorithms.  To the best of our knowledge, no other processing algorithm was published for these kernels. It can be seen that for kernels given in \cite{trofimiuk2021window} the proposed approach has slightly higher complexity compared to the window processing algorithm. However, extension of the latter  algorithm to kernels of size other than $2^\mu$ is non-obvious, while the recursive trellis algorithm can be applied to any kernel. Observe that the the proposed approach provides huge complexity reduction compared to the Viterbi algorithm. 

\begin{figure}
\includegraphics[width=0.5\textwidth]{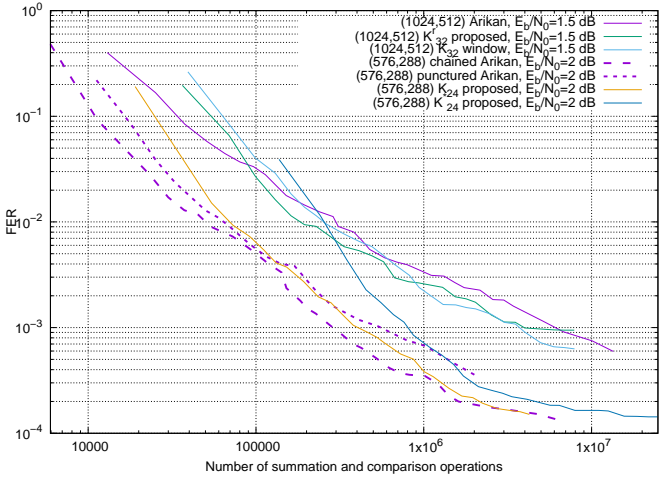}
\caption{Performance and SCL decoding complexity}
\label{fPerfCompl}
\end{figure}
\begin{figure}
\includegraphics[width=0.5\textwidth]{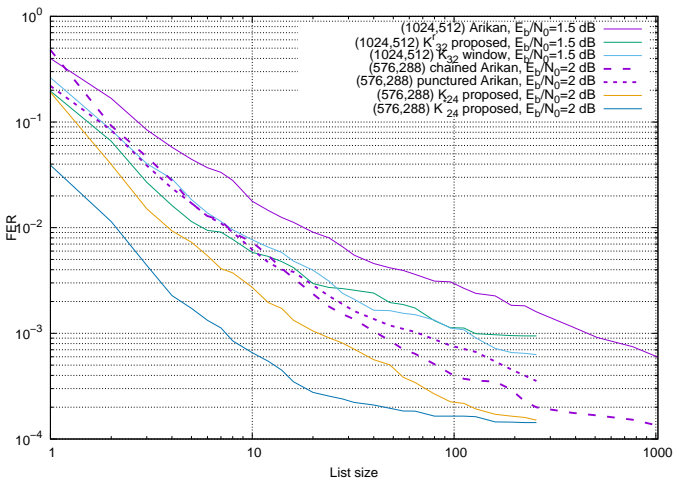}
\caption{Performance of polar codes under SCL decoding}
\label{fPerfList}
\end{figure}
Figure \ref{fPerfCompl} presents the performance and complexity of the successive cancellation list decoding algorithm with various list size for some polar subcodes \cite{trifonov2019construction}.
It can be seen that the proposed approach allows one to obtain better performance compared to the codes based on Arikan kernel with the same decoding complexity. More specifically,  $(1024,512)$ code based on  kernel $K_{32}^r$ appears to be better than the code based on Arikan kernel starting from $FER=3\cdot 10^{-2}$, which corresponds to list size 3 and  7 for $K_{32}^r$  and Arikan kernels, respectively.  Polar subcode  $(576,288)$ based on kernel $K_{24}$ outperform punctured Arikan polar subcode starting from $FER=3\cdot 10^{-3}$, which corresponds to list size 8 and 20, respectively, and approaches the performance of the chained polar subcode \cite{trifonov2018randomized}. High processing complexity of kernel $K_{24}^\ast$ allows the corresponding polar subcode  to outperform punctured polar subcode with Arikan kernel only at $FER=4\cdot 10^{-4}$, which corresponds to list size 14 and 192, respectively.

\begin{figure}
\includegraphics[width=0.5\textwidth]{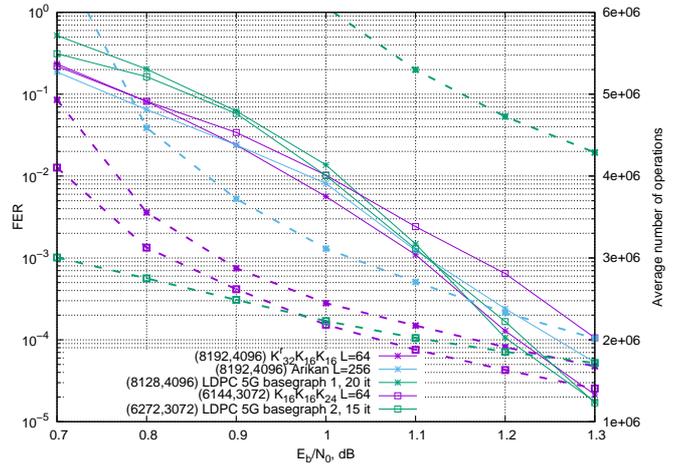}
\caption{Performance and decoding complexity of  long codes}
\label{fLongFER}
\end{figure}
Figure \ref{fLongFER} presents performance (solid lines) and average decoding complexity (dashed lines) of rate $1/2$ polar subcodes and 5G LDPC\ codes. Sequential \cite{trifonov2018score} and shuffled belief propagation \cite{zhang2005shuffled} algorithms  were used for decoding of polar subcodes and LDPC codes, respectively. Complexity is reported in terms of the number of summation and comparison operations for polar subcodes, and number of summations and calls to $\log\tanh(x/2)$ for LDPC codes. Polar subcodes were constructed using a mixture of kernels, as shown in the plot.  It can be seen that polar subcodes provide almost the same performance as LDPC codes.   At sufficiently high SNR, the decoding complexity of polar subcodes appears to be lower compared to the corresponding LDPC codes.  Observe also, that the code based on the Arikan kernel requires very large list size to obtain the performance comparable to the LDPC code. It also has higher decoding complexity compared to the code based on large kernels, and has inferior performance at high SNR.

\section{Conclusions}
In this paper a novel processing algorithm for large polarization kernels was proposed. This algorithm relies on 
 extensive reuse of the intermediate results arising in the recursive maximum likelihood decoding algorithm for the codes generated by submatrices of the considered kernel. The proposed algorithm can be applied to kernels of arbitrary dimension, and has much lower complexity compared to the Viterbi algorithm.
Derivation of a processor for a given kernel according to the proposed method involves  matrix manipulations, which should be performed once at the design time. Actual kernel processing reduces to summation and comparison of the elements of some arrays, where the indices of the operands are obtained as XOR of some pre-computed values, and partial sums given by the decisions of the SC\ algorithm.

The proposed algorithm enables polar subcodes with well-designed large kernels  to provide better performance/complexity tradeoff compared to the codes based on Arikan kernel as well as LDPC codes. However, finding large kernels with good polarization properties and low processing complexity remains, in general, an open problem.

\bibliographystyle{ieeetran}

\end{document}